\documentclass{article}

\usepackage{arxiv}

\usepackage[utf8]{inputenc} % allow utf-8 input
\usepackage[T1]{fontenc}    % use 8-bit T1 fonts
\usepackage{amsfonts}  
\usepackage{amssymb}
\usepackage{amsmath}
\usepackage{amsthm}
\usepackage{xspace}
\usepackage{xargs}
\usepackage{tikz}
\usepackage[algoruled,boxed,lined,english]{algorithm2e}
\usepackage{aliascnt}
\usepackage{array}
\usetikzlibrary{arrows}
\usepackage{cleveref}

\usepackage{todonotes}
\usepackage{lipsum}
\usepackage[normalem]{ulem}

\newtheorem{theorem}{Theorem}[section]
\newtheorem{lemma}[theorem]{Lemma}
\newtheorem{claim}[theorem]{Claim}
\newtheorem{definition}[theorem]{Definition}
\newtheorem{observation}[theorem]{Observation}
\newtheorem{proposition}[theorem]{Proposition}

\newcommand{\qedclaim}{\hfill $\diamond$ \medskip}
\newenvironment{proofclaim}{\noindent{\em Proof.}~}{\qedclaim}
\theoremstyle{definition}
\newtheorem{algo}[theorem]{Algorithm}

\makeatletter
\let\c@algocf\relax % drop existing counter
\makeatother
\newaliascnt{algocf}{theorem}
\makeatother

\newcommandx{\defparproblem}[4][]{
    \vspace{3mm}
    \noindent\fbox{
        \begin{minipage}{0.96\textwidth}
            \begin{tabular*}{\textwidth}{@{\extracolsep{\fill}}lr} {#1} \\ 
            \end{tabular*}
            {\textbf{Input:}} {#2}  \\
            {\textbf{Parameter:}} {#3}  \\
            {\textbf{Question:}} {#4}
        \end{minipage}
    }
    \vspace{4mm}
}
\newcommandx{\defsimpleproblem}[3][]{
    \vspace{3mm}
    \noindent\fbox{
        \begin{minipage}{0.96\textwidth}
            \begin{tabular*}{\textwidth}{@{\extracolsep{\fill}}lr} {#1} \\ 
            \end{tabular*}
            {\textbf{Input:}} {#2}  \\
            {\textbf{Question:}} {#3}
        \end{minipage}
    }
    \vspace{4mm}
}
\newcommand{\ie}{i.\,e.\@\xspace}

\newcommand{\motPath}{\textsc{Motivating Path}\xspace}

\newcommand{\motSG}{\textsc{Motivating Subgraph}\xspace}

\newcommand{\motSGbranch}{\textsc{Simple Motivating Subgraph}\xspace}

\newcommand{\subSum}{\textsc{Subset Sum}}
\newcommand{\kLinkage}{$\ell$-\textsc{Linkage}}
\newcommand{\exactMotKlink}{\textsc{Exact Motivating $k$-Linkage in DAG}}

\title{Time-inconsistent Planning: \\ Simple Motivation Is Hard to Find}

\author{
  Fedor V.~Fomin%\thanks{Use footnote for providing further
    \\
  Department of Informatics\\
  University of Bergen\\
  Norway \\
  \texttt{fedor.fomin@ii.uib.no} \\
   \And
 Torstein J.\,F.~Str{\o}mme%\thanks{More thanks}
 \\
  Department of Informatics\\
  University of Bergen\\
  Norway \\
  \texttt{torstein.stromme@ii.uib.no} \\
}

\begin{document}
\maketitle

\begin{abstract}
    People sometimes act differently when making decisions affecting the present moment versus decisions affecting the future only. This is referred to as time-inconsistent behavior, and can be modeled as agents exhibiting {\em present bias}. A resulting phenomenon is abandonment, which is when an agent initially pursues a task, but ultimately gives up before reaping the rewards.
    
    With the introduction of the graph-theoretic {\em time-inconsistent planning model} due to Kleinberg and Oren~\cite{Kleinberg2018timeinconsistent}, it has been possible to investigate the computational complexity of how a task designer best can support a present-biased agent in completing the task. In this paper, we study the complexity of finding a {\em choice reduction} for the agent; that is, how to remove edges and vertices from the task graph such that a present-biased agent will remain motivated to reach his target even for a limited reward. While this problem is NP-complete in general~\cite{tang2017computational,albers2019motivating}, this is not necessarily true for instances which occur in practice, or for solutions which are of interest to task designers. For instance, a task designer may desire to find the best task graph which is not too complicated.
    
    We therefore investigate the problem of finding {\em simple} motivating subgraphs. These are structures where the agent will modify his plan at most $k$ times along the way. We quantify this simplicity in the time-inconsistency model as a structural parameter: The number of branching vertices (vertices with out-degree at least 2) in a minimal motivating subgraph.
    
    Our results are as follows: We give a linear algorithm for finding an optimal motivating path, i.\,e.~when $k=0$. On the negative side, we show that finding a simple motivating subgraph is NP-complete even if we allow only a single branching vertex --- revealing that simple motivating subgraphs are indeed hard to find. However, we give a pseudo-polynomial algorithm for the case when $k$ is fixed and edge weights are rationals, which might be a reasonable assumption in practice.
\end{abstract}

% keywords can be removed

%%  %%%%%%%%%%%%%%%%%%%%%%%%%%%%%%%%%%%%%%%%%%%%%%%%%%%%%
 %  %%%%%%%%%%%%%%%%%%%%%%%%%%%%%%%%%%%%%%%%%%%%%%%%%%%%%
 %% %%%%%%%% END:     TITLE, ABSTRACT, METADATA  %%%%%%%%
 %  %%%%%%%%%%%%%%%%%%%%%%%%%%%%%%%%%%%%%%%%%%%%%%%%%%%%%
%%  %%%%%%%%%%%%%%%%%%%%%%%%%%%%%%%%%%%%%%%%%%%%%%%%%%%%%

    %%%%%%%%%%%%%%%%%%%%%%%%%%%%%%%%%%%%%%%%%%%%%%%%%%%%%  %%
    %%%%%%%%%%%%%%%%%%%%%%%%%%%%%%%%%%%%%%%%%%%%%%%%%%%%%  %
    %%%%%%%%%%%% BEGIN:   PAPER   %%%%%%%%%%%%%%%%%%%%%%% %%
    %%%%%%%%%%%%%%%%%%%%%%%%%%%%%%%%%%%%%%%%%%%%%%%%%%%%%  %
    %%%%%%%%%%%%%%%%%%%%%%%%%%%%%%%%%%%%%%%%%%%%%%%%%%%%%  %%

      %%%%%%%%%%%%%%%%%%%%%%%%%%%%%%%%%%%%%%%%%%%%%%%%%%%%%  %%
      %%%%%%%%%%%%%%%%%%%%%%%%%%%%%%%%%%%%%%%%%%%%%%%%%%%%%  %
      %%%%%%%%%%%% BEGIN:  INTRO   %%%%%%%%%%%%%%%%%%%%%%%% %%
      %%%%%%%%%%%%%%%%%%%%%%%%%%%%%%%%%%%%%%%%%%%%%%%%%%%%%  %
      %%%%%%%%%%%%%%%%%%%%%%%%%%%%%%%%%%%%%%%%%%%%%%%%%%%%%  %%
    
      %!TEX root = motivate.tex
\section{Introduction}
\label{sec:intro}

Time-inconsistent behavior is a theme attracting great attention in behavioral economics and psychology. The field investigates questions such as why people let their bills go to debt collection, or buy gym memberships without actually using them. More generally, inconsistent behavior over time occurs when an agents makes a multi-phase plan, but does not follow through on his initial intentions despite circumstances remaining essentially unchanged. Resulting phenomenons include procrastination and abandonment.

A common explanation for time-inconsistent behavior is
the notion of \emph{present bias}, which states that agents give undue salience to events that are close in time and/or space. This idea was described mathematically already in 1937 when Paul Samuelson \cite{Samuelson1937} introduced the discounted-utility model, which has since been refined in different versions~\cite{Laibson1994}.
%
%
%This idea was described mathematically already in 1937 when Paul Samuelson \cite{Samuelson1937} introduced the discounted-utility model, which has since been refined in different versions. In the 60's, Phelps and Pollak \cite{PhelpsPollak1968} used a model where the utility of future events are prescribed a discount factor $\beta\delta^t$, whereas immediate events have no such discount factor. In their scenario, $\beta$ is common to all non-immediate events ($0 \leq \beta \leq 1$), whereas $\delta^t$ becomes exponentially more influential the further away an event is in time ($0 \leq \delta \leq 1$, $t \geq 1$). This model is today known as \emph{quasi-hyperbolic discounting}, a term coined by David Laibson in his 1994 thesis \cite{Laibson1994}. We remark that when $\beta = 1$, this model is equal to standard exponential discounting used in the discounted-utility model%
%
%\footnote{Quasi-hyperbolic discounting is shown to be a better real world predictor than the discounted-utility model, and there is also empirical support for this model from psychology. For instance, McClure et. al showed by using brain imaging that separate neural systems are in play when humans value immediate and delayed rewards \cite{McClure2004}. This suggests that distinguishing the valuation schemes of immediate and future events is crucial, even if the delay to the future is short (the role of $\beta$). However there are also many known psychological phenomenon about time-inconsistent behavior that quasi-hyperbolic discounting does not capture \cite{Frederick2002}.}. 
%
%

%
%On the other extreme,
George Akerlof describes in his 1991 lecture  \cite{akerlof1991procrastination} an even simpler mathematical model; here, the agent simply has a {\em salience factor} causing immediate events to be emphasized more than future events. He goes on to show how even a very small salience factor in combination with many repeated decisions can lead to arbitrary large extra costs for the agent. This salience factor also has support from psychology, where McClure et.~al showed by using brain imaging that separate neural systems are in play when humans value immediate and delayed rewards~\cite{McClure2004}.%, even if the delay is relatively short.
\footnote{Note that quasi-hyperbolic discounting (discussed in~\cite{Laibson1994,McClure2004}) can be seen as a generalization of both Samuelson's discounted-continuity model~\cite{Samuelson1937} and Akerlof's salience factor~\cite{akerlof1991procrastination}. There has been some empirical support for this model; however there are also many known psychological phenomena about time-inconsistent behavior it does not capture~\cite{Frederick2002}.}
%that are consistent with quasi-hyperbolic discounting, but where essentially $\delta = 1$, \ie he cares only about the difference in valuation of immediate versus non-immediate events, or the \emph{salience factor}
%
%\todo[inline]{Decide: bias factor vs salience factor. Akerlof: ``extra salience''. Kleinberg: ``present bias''. Tang/Albers: $\beta$ called ``discount parameter''. But ``present bias'' also appear to be used about $\beta\delta^t$ as a whole, or more informally, in other parts of literature.} %% Settlede on salience factor

%Akerlof goes on to show how even a very small salience factor in combination with many repeated decisions can lead to arbitrary large extra costs for the agent. In the current paper, we for simplicity let $b$ denote the salience factor ($b = \beta^{-1}$), and we will apply $b$ to the valuation of immediate events, rather than applying $\beta$ to the valuation of non-immediate events. 

% For instance, it has been observed that beyond a certain time horizon, the discount factor is essentially the same - suggesting that the $\delta^i$ term is somewhat misleading \cite{Frederick2002TimeDiscountingCriticalReview}, or should at least be much larger than $\beta$.}.

In 2014, Kleinberg and Oren~\cite{Kleinberg2018timeinconsistent} introduced a graph-theoretic model which elegantly captures the salience factor and scenarios of Akerlof. In this framework, where the agent is maneuvering through a weighted directed acyclic graph, it is possible to model many interesting situations. We will provide an example here.
\subsection{Example}%
\begin{figure}
    \centering
    \includegraphics[]{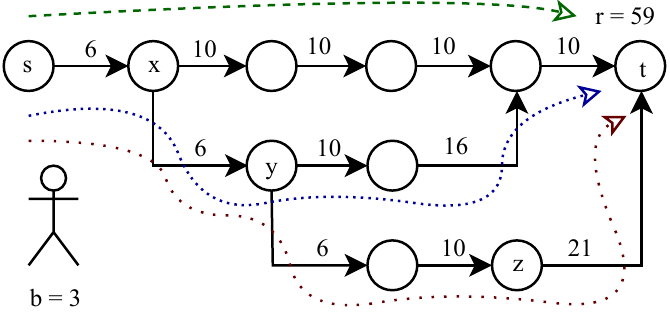}
    \caption{Acyclic digraph illustrating the ways in which Bob can distribute his efforts in order to complete the course. The upper path (green dashed line) is Bob's initial plan requiring the least total effort. The middle path (blue, narrowly dotted line) is the plan which appears better when at vertex $x$, and lower path (red, widely dotted line) is the plan he ultimately changes to at vertex $y$.} \label{fig:example1}
\end{figure}%
The student Bob is planning his studies. He considers taking a week-long course whose passing grade is a reward he quantifies to be worth $r = 59$. And indeed, Bob discovers that he can actually complete the course incurring costs and effort he quantifies to be only $46$ --- if he works evenly throughout the week (the upper path in Figure~\ref{fig:example1}). Bob will reevaluate the cost every day, and as long as he perceives the cost to be at most equal to the reward, he will follow the path he finds to have the lowest cost.

The first day of studies incurs a cost of $6$ for Bob due to some mandatory tasks he needs to do that day. But because Bob has a salience factor $b = 3$, he actually perceives the cost of that day's work to be $18$, and of the course as a whole to be $58$ ($18 + 10 + 10 + 10 + 10$). The reward is even greater, though, so Bob persists to the next day.

When the second day of studies is about to start, Bob quasi-subconsciously makes the incorrect judgment that reducing his studies slightly now is the better strategy. He then changes his plan to the middle path in Figure~\ref{fig:example1}. In terms of our model, the agent Bob standing at vertex $x$ reevaluates the original plan (the upper path) to now cost $3\cdot10 + 10 + 10 + 10 = 60$, whereas the middle path is evaluated to only cost $3 \cdot 6 + 10 + 16 + 10 = 54$. He therefore chooses to go on with the plan that postpones some work to later, incurring a small extra cost to be paid at that time.

On the third day Bob finds himself at vertex $y$, and is yet again faced with a choice. The salience factor, as before, cause him to do less work in the present moment at the expense of more work in the future. He thus changes his plan to the lower path of Figure~\ref{fig:example1}. However, it turns out that the choice was fatal --- on the last day of the course (at vertex $z$), Bob is facing what he perceives to be a mountain of work so tall that it feels unjustified to complete the course; he evaluates the cost to be $3 \cdot 21 = 63$, strictly larger than the reward. He gives up and drops the course.

Because Bob abandons the task in our example above, we say that the graph in Figure~\ref{fig:example1} is not {\em motivating}. A natural question is to ask what we can do in order to make it so.

An easy solution for making a model motivating is to simply increase the reward. By simulating the process, it is also straightforward to calculate the minimum required reward to obtain this. However, it might be costly if we are the ones responsible for purchasing the reward, or even impossible if it is not for us to decide. A more appealing strategy might therefore be to allow the agent to only move around in a subgraph of the whole graph; for instance, if the lower path did not exist in our example above, then the graph would actually be motivating for Bob.\footnote{Removing edges and/or vertices private to the lower path is also the {\em only} option for how to make the example graph motivating; the upper path must be kept in its entirety, otherwise the agent will not be motivated to move from $s$ to $x$; the middle path must also be kept in its entirety, otherwise the agent will give up when at $x$.} Finding such a subgraph is a form of {\em choice reduction}, and can be obtained by introducing a set of rules the agent must follow; for instance deadlines.

The aim of the current paper is not, however, to delve into the details of any particular scenario, but rather to investigate the formal underlying graph-theoretic framework. In this spirit, Kleinberg and Oren~\cite{Kleinberg2018timeinconsistent} show that the structure of a minimal motivating subgraph is actually quite restricted, and ask whether there is an algorithm finding such subgraphs. Unfortunately, Tang et al.~\cite{tang2017computational} and  Albers and Kraft~\cite{albers2019motivating} independently proved that this problem is NP-complete in the general case. However, this does not exclude the existence of polynomial time algorithms for more restricted classes of graphs, or algorithms where the exponential blow-up occurs in parameters which in practical instances are small. This is what we investigate in the current paper; specifically, we look at restricting the number of {\em branching vertices} (vertices with out-degree at least 2) in a minimal motivating subgraph. This parameter can also be understood as the number of times a present-biased agent changes his plan.
%
%As observed in~\cite{Kleinberg2018timeinconsistent}, it is possible to remove all edges the agent never plans to traverse at any point throughout his journey, without affecting his behaviour; what remains can be seen as the ``core'' of the motivating subgraph. Branching vertices in the core are exactly those points where the agent changes his plan. It is clear that all minimal motivating subgraphs also are their own cores, and thus possess this property.
%
%Looking for minimal motivating subgraphs with at most $k$ branching vertices hence makes sense: It corresponds to finding a motivating structure where the agent change his plan at most $k$ times.

Before we present our results, let us introduce the model more formally.

\subsection{Formal model}

\begin{table}[]
    \centering
    \def\arraystretch{1.3}
    \begin{tabular}{|c|m{11cm}|}
        \hline
        %        $\mathbb{R}$  & The set of {\em real numbers}. \\ \hline
        %        $\mathbb{Z}$  & The set of {\em integers}. \\ \hline
        $A_{\varphi}$ & For a set $A$ and a constraint $\varphi:A \to \{\texttt{T}, \texttt{F}\}$, the set {\em $A$ constrained to $\varphi$} denotes the elements of $A$ that satisfy $\varphi$, \ie{}~$\{a \in A \mid \varphi(a) = \texttt{T}\}$. For example, $\mathbb{Z}_{\geq 0}$ indicates the set of all non-negative integers. \\ \hline
        $[x]$         & For $x \in \mathbb{Z}_{\geq 1}$, $[x]$ is the set $\{1,2,\ldots,x\}$. \\ \hline
        $f|_A$        & For a function $f:B \to C$ and a set $A \subseteq B$, the function {\em $f$ restricted to $A$} is a function $f|_A:A \to C$ such that for every $a \in A$, $f|_A(a) = f(a)$. \\ \hline
        $G[A]$        & For a directed graph $G$ and vertex set $A \subseteq V(G)$, the {\em induced subgraph} $G[A]$ is the graph where $V(G[A]) = A$ and $E(G[A]) = E(G) \cap A \times A$. \\ \hline
    \end{tabular}
    \vspace*{2mm}
    \caption{Summary of notation.}
    \label{tab:notation}
\end{table}
%\begin{wrapfigure}{r}{0.3\textwidth}
%    \footnotesize
%    \begin{framed}
%        \subsubsection*{Common terms}
%        \begin{itemize}
%            \item[$\mathbb{R}$] The set of real numbers.
%            \item[$\mathbb{Z}$] The set of integers.
%            \item[$A_{\geq x}$] The set $A$ restricted to elements greater or equal to $x$.
%            \item[{$[k]$}] The set $\{1, 2,\ldots,k\}$.
%            \item[$f|_A$] The restriction of function $f$ to domain $A$.
%            \item[{$G[A]$}] The subgraph of $G$ induced on $A$.
%        \end{itemize}
%    \end{framed}
%\end{wrapfigure}

We here present the model due to Kleinberg and Oren~\cite{Kleinberg2018timeinconsistent}. Formally, an instance of the {\em time-inconsistent planning model} is a 6-tuple $M = (G, w, s, t, r, b)$ where:
\begin{itemize}
    \item $G = (V(G), E(G))$ is a directed acyclic graph called a {\em task graph}. $V(G)$ is a set of elements called {\em vertices}, and $E(G) \subseteq V(G)\times V(G)$ is a set of directed {\em edges}. The graph is {\em acyclic}, which means that there exists an ordering of the vertices called a {\em topological order} such that, for each edge, its first endpoint comes strictly before its second endpoint in the ordering.  Informally speaking, vertices represent states of intermediate progress, whereas edges represent possible actions that transitions an agent between states.
    \item $w : E(G) \to \mathbb{R}_{\geq 0}$ is a function assigning non-negative weight to each edge. Informally speaking, this is the cost incurred by performing a certain action.
    \item $s \in V(G)$ is the start vertex.
    \item $t \in V(G)$ is the target vertex.
    \item $r \in \mathbb{R}_{\geq 0}$ is the reward.
    \item $b \in \mathbb{R}_{\geq 1}$ is the agent's salience factor.\footnote{While we in the current paper use $b$ as the salience factor as introduced in~\cite{Kleinberg2018timeinconsistent}, the literature about time-inconsistent planning commonly use the term $\beta = b^{-1}$ instead, which, while slightly more convoluted to work with for our purposes, seamlessly integrates with the quasi-hyperbolic discounting model. An artifact of our definition in the current paper is that the reward is not scaled by $b$ when the agent is one leg away; however, both algorithms and hardness proofs can be adapted to account for this technical difference.}
\end{itemize}

An agent with salience factor $b$ is initially at vertex $s$ and can move in the graph along edges in their designated direction. The agent's task is to reach the target $t$, at which point the agent can pick ut a reward worth $r$. We can usually assume that there is at least one path from $s$ to each vertex, and at least one path from each vertex to $t$, as otherwise these vertices are of no interest to the agent.

When standing at a vertex $u$, the agent evaluates (with a present bias) all possible paths from $u$ to $t$. In particular, a $u$-$t$ path $P \subseteq G$ with edges $e_1, e_2, \ldots,e_p$ is evaluated by the agent standing at $u$ to cost $\zeta_M(P) = b\cdot w(e_1) + \sum^{p}_{i=2}w(e_i)$. We refer to this as the {\em perceived} cost of the path. For a vertex $u$, its perceived cost to the target is the minimum perceived cost of any path to $t$, $\zeta_M(u) = \min \{\zeta_M(P) \mid P \text{ is a }u\text{-}t \text{ path}\}$. If the perceived cost from vertex $u$ to the target is strictly larger than the reward, $\zeta_M(u) > r$, then an agent standing there {\em abandons} the task. Otherwise, he will (non-deterministically) pick one of the paths which minimize perceived cost of reaching $t$, and traverse its first edge. This repeats until the agent either reach $t$ or abandons the task.

If every possible route chosen by the agent will lead him to $t$, then we say that the model instance is {\em motivating}. If the model instance is clear from the context, we take that the {\em graph} is motivating to mean the same thing, and we may drop the subscript $_M$ in the notation.

%\begin{definition}[Induced function] \todo{Is there some standard notation for this? Move to preliminaries.}
%For two sets $A$, $B$, a function $f : A \to B$ and a subset $A' \subseteq A$, we let $f[A'] : A' \to B$ denote the function {\em $f$ induced on $A'$} where $f[A'](a) = f(a)$ for all $a \in A'$.
%\end{definition}

\begin{definition}[Motivating subgraph]
%A subgraph $G' \subseteq G$ is a graph where $V(G') \subseteq V(G)$ and $E(G') \subseteq E(G)$.
If $G'$ is a subgraph of $G$ belonging to a time-inconsistent planning model $M = (G, w, s, t, r, b)$, then we call $G'$ a {\em motivating subgraph} if $G'$ contains $s$ and $t$ and $M' = (G', w|_{E(G')}, s, t, r, b)$ is motivating%
%\footnote{Note that $w|_{E(G')}$ refers to the function $w$ restricted to the domain $E(G')$.}%
.
\end{definition}

In the current paper, we investigate the problem of finding a simple motivating subgraph. In order to quantify what we mean by {\em simple}, we first provide the definition of a {\em branching vertex}:

\begin{definition}[Branching vertex]
    The {\em out-degree} of a vertex $v$ in a directed graph $G$ is the number of edges in $G$ that have $v$ as its first endpoint. We  say that $v$ is a \emph{branching vertex}, if its out-degree is at least two.  %has at least two incident out-going arcs. 
\end{definition}
 
In its most general form, we will investigate the following problem:

 \defsimpleproblem{\motSGbranch{}}{A time-inconsistent planning model $M = (G, w, s, t, r, b)$, and a non-negative integer $k \in \mathbb{Z}_{\geq 0}$.} {Does there exist a motivating subgraph  $G' \subseteq G$ with at most $k$ branching vertices?}

We observe that if $b = 1$, then the problem is merely a question of finding the shortest path in a graph, which can be done in linear time by dynamic programming on the topological order. When $b$ approaches infinity, then the problem becomes equivalent to finding the min-max path from $s$ to $t$, which can be solved similarly. Furthermore, when $r$ is $0$, the problem boils down to reachability in the graph where only $0$-weight edges are kept. We will henceforth assume that $b$ is a constant strictly greater than $1$, and that $r$ is strictly positive.
 
% \medskip\noindent\textbf{Previous work.}
 \subsection{Previous work.}%
 Time-inconsistent behavior is a field with a long history in behavioral economics, see~\cite{akerlof1991procrastination,Frederick2002,o1999doing}. Kleinberg and Oren~\cite{Kleinberg2018timeinconsistent} introduce the graph-theoretic model used in this paper, give structural results concerning how much extra cost the salience factor can incur for an agent, and how many values the salience factor can take which will lead the agent to follow distinct paths. They also raise the issue of motivating subgraphs, and give a useful characterization of the minimal among them which we later will see.
 % the cost ratio between the path taken by the agent and the optimal path is bounded in graphs excluding a certain minor.
 % if there are many many agents each with their own distinct bias b, they will follow at most O(n^2) different paths through the graph
 % 
% \begin{proposition}[{\cite[Theorem~5.1]{Kleinberg2018timeinconsistent}}]
%     \label{prop:motsubgraphs_are_sparse2}
%     If $G'$ is a minimal motivating subgraph for a time-inconsistent planning model $M = (G, w, s, t, r, b)$, then $G'$ contains an $s$-$t$ path $P'$ with the properties that:
%     \begin{itemize}
%         \item[(i)] Every edge of $G'$ is either part of $P'$ or lies on another path $Q'$ which is such that it intersects $P'$ exactly in its endpoints. In other words, each edge can be reached from $s$ and can be traversed to reach $t$.
%         \item[(ii)] Every vertex of $G'$ has at most one outgoing edge that is not in $P'$.
%     \end{itemize}
% \end{proposition}\todo{remove proposition either here or later}
%%

Tang et al.~\cite{tang2017computational} refine the structural results concerning extra costs caused by present bias. Furthermore, they show that finding motivating subgraphs is NP-complete in the general case by a reduction from $\textsc{3-Sat}$. They also show hardness for a few variations of the problem where intermediate rewards can be placed on vertices, and give a 2-approximation for these versions.

Albers and Kraft~\cite{albers2019motivating} independently show that finding motivating subgraphs is NP-complete in the general case by a reduction from \kLinkage{} problem in acyclic digraphs. Furthermore, they show that the approximation version of the problem (finding the smallest $r$ such that a motivating subgraph exists) cannot be approximated in polynomial time to a ratio of $\sqrt{n}/3$ unless P = NP; but a $1 + \sqrt{n}$ -approximation algorithm exists. They also explore another variation of the problem with intermediate rewards, which they show to be NP-complete.

There have been work on variations on the model and problem where the designer is free to raise edge costs~\cite{albers2017penalties}, where the agents are more sophisticated~\cite{kleinberg2016planning}, exhibit multiple biases simultaneously~\cite{kleinberg2017planning}, or where the salience factor varies~\cite{albers2017price,gravin2016procrastination}.
 
\subsection{Our contribution}
 
% \medskip\noindent\textbf{Our results.}
We prove two main results about  the complexity of \motSGbranch. In short, our results can be summarized as follows. While solvable in linear time when $k = 0$, \motSGbranch{} is NP-hard already for $k=1$. However, the hardness reduction strongly exploits constructions with exponentially large (or small) edge weights. In a more realistic scenario, when the costs are bounded by some polynomial of the number of tasks, the problem is solvable in polynomial time for every fixed $k$.

More precisely, our first result is the following dichotomy theorem. 
\begin{theorem}\label{theoremNPdich}\motSGbranch{} 
  is solvable in polynomial time for $k=0$, and is NP-complete for any $k\geq 1$. 
 \end{theorem}
 
 The reduction we use to prove NP-completeness of the problem in Theorem~\ref{theoremNPdich} is from \subSum, which is weakly NP-complete. Thus, the numerical parameter $W$ in our reduction --- the sum of edge weights when scaled to integer values --- 
 %
 %the ratio of the largest and the smallest edge weights\todo{Is this true? E.g. what if all numbers are very large primes? Ratio of smallest to largest can be constant; can \subSum be solved quickly then? Eg W = 9221713, \{3073891, 3073901, 3073921\}}, 
 %
 is exponential in the number of vertices in the graph. 
 A natural question is hence whether  \motSGbranch{} can be solved by a pseudo-polynomial  algorithm, i.\,e.~an algorithm which runs in polynomial time when $W$ 
 is bounded by some polynomial of the size of the graph. 
  Unfortunately, this is highly unlikely. A closer look at the  NP-hardness 
  proof of  Tang et al.~\cite{tang2017computational} of finding a motivating subgraph, reveals that instances created in the reduction from 3-SAT have weights that depend only on the constant $b$. 
  
  On the other hand,  in the reduction of  Tang et al.~\cite{tang2017computational}  the number of branchings in their potential solutions grows linearly with the number of clauses in the 3-SAT instance. This leaves a possibility that  when $W$ is bounded by a polynomial of  the size of the input graph $G$  and $k$ is a constant, then  \motSGbranch{} is solvable in polynomial time. Theorem~\ref{theoremXP} confirms that this is exactly the case. (In this theorem we assume that all edge weights are integers --- but since scaling the reward and all edge weights by a common constant is trivially allowed, it also works for rationals.)

 \begin{theorem}\label{theoremXP}
   \motSGbranch{} 
  is solvable in time $(|V(G)|\cdot W)^{\mathcal{O}(k)}$ whenever all edge weights are integers and $W$ is the sum of all weights.
\end{theorem}

Theorem~\ref{theoremXP}  naturally leads to another question, whether \motSGbranch{}
is solvable in time  $(|V(G)|\cdot W)^{\mathcal{O}(1)}\cdot f(k)$ for some function $f$ of $k$ only. Or in other words, whether the problem is fixed-parameter tractable parameterized by $k$ when  $W$ is encoded unary?    We observe that Albers and Kraft's~\cite{albers2019motivating} hardness proof, reducing from the \kLinkage{} problem in acyclic digraphs, combined with the hardness result of Slivkins \cite{Slivkins10}, implies the following theorem. 

%\todo[inline]{Theorem~\ref{theoremW1} does not seem to follow from what we say. The proof of Albers and Kraft \cite{Albers2018} needs disjoint pairs of terminals. Does the proof of Slivkins work in this case? }

 \begin{theorem}\label{theoremW1}
   Unless  $\operatorname{FPT}$  $=$ $\operatorname{W[1]}$, there is no algorithm solving \motSGbranch{} 
  in time $(|V(G)|\cdot W)^{\mathcal{O}(1)}\cdot f(k)$ for any function $f$ of $k$ only. %The lower bound holds even when $W\in  \mathcal{O}( k)$.
\end{theorem}

%
%\todo[inline]{(Torstein) \sout{I think it might be possible to modify the reduction of Albers and Kraft to have weights be bounded by a constant and instead blow up the size of the graph by a polynomial function of b/k.} Edit: I do no longer believe such a modification is possible. If it is, the technique can not be applicable to the reduction from subset sum.}

%
%Albers and Kraft \cite{Albers2018}: $W$ is a function of disjpaths, the number of branches number of paths. 
%
% \cite{tang2017computational}: weight bounded but number of branches ounbounded. 
% 
% 
% \cite{Slivkins10}

  %%  %%%%%%%%%%%%%%%%%%%%%%%%%%%%%%%%%%%%%%%%%%%%%%%%%%%%%
   %  %%%%%%%%%%%%%%%%%%%%%%%%%%%%%%%%%%%%%%%%%%%%%%%%%%%%%
   %% %%%%%%%%%%%% END :   INTRO   %%%%%%%%%%%%%%%%%%%%%%%%
   %  %%%%%%%%%%%%%%%%%%%%%%%%%%%%%%%%%%%%%%%%%%%%%%%%%%%%%
  %%  %%%%%%%%%%%%%%%%%%%%%%%%%%%%%%%%%%%%%%%%%%%%%%%%%%%%%
  
      %%%%%%%%%%%%%%%%%%%%%%%%%%%%%%%%%%%%%%%%%%%%%%%%%%%%%  %%
      %%%%%%%%%%%%%%%%%%%%%%%%%%%%%%%%%%%%%%%%%%%%%%%%%%%%%  %
      %%%%%%%%%%%% BEGIN:  DICHOTOMY   %%%%%%%%%%%%%%%%%%%% %%
      %%%%%%%%%%%%%%%%%%%%%%%%%%%%%%%%%%%%%%%%%%%%%%%%%%%%%  %
      %%%%%%%%%%%%%%%%%%%%%%%%%%%%%%%%%%%%%%%%%%%%%%%%%%%%%  %%
      
      %!TEX root = motivate.tex
\section{A dichotomy (proof of Theorem~\ref{theoremNPdich})}
\label{sec:paths}

In this section we prove Theorem~\ref{theoremNPdich}. Recall that the theorem states that 
\motSGbranch{} 
  is solvable in polynomial time for $k=0$, and is NP-complete for any $k\geq 1$. 
We split the proof into two subsections. The first subsection contains a polynomial time algorithm solving \motSGbranch{} for $k=0$ (Lemma~\ref{lem:motPath}). The second subsection proves that the problem is NP-complete for every  $k\geq 1$ (Lemma~\ref{lem:motNP}). 

\subsection{\motPath}
Any connected graph without branching vertices is a path. Thus, we refer to the variant of \motSGbranch{} with  $k=0$ as to the \motPath{} problem. We solve this with an algorithm that is very similar to the classical linear time algorithm computing a shortest path in a DAG.

%
%
%\defsimpleproblem{\motPath{} (\mP{})}%
%{A directed acyclic graph $G$ with non-negative edge weights $w : E(G) \to \mathbb{R}_{\geq 0}$; vertices $s, t \in V(G)$; a bias factor $b \in \mathbb{R}_{\geq 1}$; and a reward $r \in \mathbb{R}_{\geq 0}$.}%
%{Does there exist a path $P \subseteq G$ with $s, t \in V(P)$ where an agent with bias factor $b$ starting in $s$ will be motivated to reach $t$ by a reward $r$?}

We hereby present Algorithm~\ref{alg:motpath}, which solves the \motPath{} problem in $\mathcal{O}(|V(G)| + |E(G)|)$ time. In fact, our algorithm will solve the problem of finding the minimum length such path, if one exists.

% \begin{algorithm}
%     \caption{\motPath{}}\label{alg:motpath}
%     \KwIn{An instance of \motPath}
%     \KwOut{A shortest path $P \subseteq G$ witnessing a yes-instance; or \texttt{None} if no path exists.}
 
%     For each vertex $u \in V(G)$, let $d_u \gets \infty$\;
%     For each vertex $u \in V(G)$, let $p_u \gets \texttt{None}$\;
%     For the target $t$, let $d_t \gets 0$\;
%     \For{each vertex $u \in V(G)$ in reverse topological order}{
%         \For{each out-neighbour $v \in N(u)$}{
%             \If{$b\cdot w(uv) + d_v \leq r$ \textnormal{\textbf{and}} $w(uv) + d_v < d_u$}{
%                 $d_u \gets w(uv) + d_v$\;
%                 $p_u \gets v$\;
%             }
%         }
%     }
%     \eIf{$p_s = \textnormal{\texttt{None}}$}{
%         \Return \texttt{None}\;
%     }
%     {
%         $V(P) \gets \{ s \}$\;
%         $E(P) \gets \{ \}$\;
%         $u \gets s$\;
%         \While{$p_u \neq \textnormal{\texttt{None}}$}{
%             $V(P) \gets V(P) \cup \{p_u\}$\;
%             $E(P) \gets E(P) \cup \{up_u\}$\;
%             $u \gets p_u$\;
%         }
%         \Return $(V(P), E(P))$;
%     }
% \end{algorithm}

\begin{algorithm}
    \caption{\motPath{}}\label{alg:motpath}
    \KwIn{An instance of \motPath}
    \KwOut{The length of a shortest motivating $s$-$t$ path witnessing a yes-instance; or $\infty$ if no motivating path exists.}
 
    For each vertex $u \in V(G)$, let $d_u \gets \infty$\;
    For the target $t$, let $d_t \gets 0$\;
    \For{each vertex $u \in V(G)$ in reverse topological order}{
        \For{each out-neighbour $v \in N(u)$}{
            \If{$b\cdot w(uv) + d_v \leq r$ \textnormal{\textbf{and}} $w(uv) + d_v < d_u$}{
                $d_u \gets w(uv) + d_v$\;
            }
        }
    }
    \Return $d_s$
\end{algorithm}

\begin{lemma}\label{lem:motPath}
    The \motPath{} problem can be solved in linear time.
\end{lemma}
\begin{proof}
    We prove that Algorithm~\ref{alg:motpath} is correct. We assume that every vertex has a path to $t$ in $G$ (otherwise we can simply remove it), hence $t$ will come last in the topological order. For every $u \in V(G)$, we claim that $d_u$ holds the minimum length of a motivating path from $u$ to $t$. We observe that our base, $u = t$, is correct since $d_t = 0$.
    
    Consider some vertex $u$. Because the vertices are visited in reverse topological order, all out-neighbors of $u$ are already processed, and hold by the induction hypothesis a correct value. An agent standing at vertex $u$ is motivated to move to the next vertex $v$ in a path $P$ if $b \cdot w(uv) + \textsc{dist}_P(v, t) \leq r$. Hence, if the condition holds, prepending a motivating path from $v$ to $t$ with $u$ will also yield a motivating path. By choosing the shortest total length among all feasible candidates for the next step, the final value is in accordance with our claim.
    
    Finally, we observe that the runtime is correct. Assuming an adjacency list representation of the graph, a topological sort can be done in linear time, and the algorithm process each vertex once and touches each edge once.
\end{proof}

We remark that the graph produced by Algorithm~\ref{alg:motpath} is a $(1 + \sqrt{n})$ -approximation to the general motivating subgraph problem. This follows because the approximation algorithm of Albers and Kraft~\cite{albers2019motivating} with the stated approximation ratio always produce a path; Algorithm~\ref{alg:motpath} will on the other hand find the optimal path, and is hence at least as good.

%\todo[inline]{Check if this follows from earlier results}

      %!TEX root = motivate.tex
\subsection{Hardness of allowing branches}
\label{sec:hardbranchings}

Deciding whether there exists a \emph{path} which will motivate a biased agent to reach the target turns out to be easy, but we already know by the result of Tang et al.~\cite{tang2017computational} and Albers and Kraft~\cite{albers2019motivating} that finding a motivating \emph{subgraph} in general is NP-hard. In both  reductions of Tang et al.  and  of Albers and Kraft, the feasible solutions to the reduced instance have a ``complicated'' structure in the sense that they contain many branching vertices. A natural question is whether it could be easier to find motivating subgraphs whose structure is simpler, as in the case of paths.

A first question might be whether \motSGbranch{} is {\em fixed parameter tractable} (FPT) parameterized by $k$, the number of branching vertices. Unfortunately, this is already ruled out by the reduction of Albers and Kraft, since they reduce from the W[1]-hard \kLinkage{} problem for acyclic digraphs --- the number of branchings in their feasible solutions is linear in $\ell$.

As the \kLinkage{} problem can be solved in time $n^{f(\ell)}$ for acyclic digraphs~\cite{bang2008digraphs}, their reduction does not rule out an XP-algorithm. In this section we show by a reduction from \subSum{} that \motSGbranch{} is actually NP-hard even for $k = 1$.

%
%
%With the perspective of parameterized complexity we can quantify how complicated the structure of the solution is by counting how many \emph{branchings} (\ie vertices of out-degree more than one) it contains. We define the following parameterized problem:
%
%\defparproblem{\motSGbranch{} (\mSGb)}{A directed acyclic graph $G$ with non-negative edge weights $w : E(G) \to \mathbb{R}_{\geq 0}$; vertices $s, t \in V(G)$; a bias factor $b \in \mathbb{R}_{> 1}$; a reward $r \in \mathbb{R}_{>0}$; and a non-negative integer $k \in \mathbb{Z}_{\geq 0}$.}{$k$}{Does there exist a graph $G' \subseteq G$ where at most $k$ vertices of $G'$ have more than one outgoing edge, with $s, t \in V(G')$ such that an agent with bias factor $b$ starting in $s$ will be motivated to reach $t$ by a reward $r$?}

\begin{lemma}\label{lem:motNP}
     \motSGbranch{}   is  NP-complete for $k \geq 1$.
\end{lemma}
\begin{proof} Deciding whether a given graph is motivating  can be done in polynomial time by simply checking whether it is possible for the agent to reach a vertex where the perceived cost is greater than the reward. This implies the membership of  \motSGbranch{} in NP.
%To check whether a subgraph $G'$  is motivating, we simulate the actions of the agent on it. At every step the computation of a shortest \todo{reference}

To prove NP-compleness, we reduce from the classical NP-complete problem \subSum{}~\cite{Karp72}. 

\defsimpleproblem{\subSum{}}{A set of non-negative integers $X = \{x_1, x_2, \ldots, x_n\} \subseteq \mathbb{Z}_{\geq 0}$ and a target $W \in \mathbb{Z}_{\geq 0}$.}{Does there exists a subset $X' \subseteq X$ such that its elements sums to $W$?}

The reduction is  described in the form of   Algorithm~\ref{red:subSumToMotSGBranch} and an example is given in Figure~\ref{fig:subSumToMotSGBranch}. The soundness of the reduction is proved in the following claim.

\begin{algorithm}
    \DontPrintSemicolon
    \caption{Reduction from \subSum{} to \motSGbranch{} ($k=1$)}
    \label{red:subSumToMotSGBranch}
    \KwIn{An instance $I = (X = \{x_1, x_2, \ldots, x_n\}, W)$ of \subSum{}; and any salience factor $b \in \mathbb{R}_{> 1}$.}
    \KwOut{An instance $I' = (G, w, s, t, b, r, k)$ of \motSGbranch{} with $k=1$ and salience factor $b$.}
    $V(G) \gets$ $\{s, a_0, a_1, a_2, a_3, t\} \cup \{c_1, c_1^*, c_2, c_2^*, \ldots, c_{n}, c_{n}^*\} \cup \{c_{n+1}, c_{n+2}\}$\;
    $E(G) \gets \{sa_0, a_0a_1, a_1a_2, a_2a_3, a_3t\} \cup \{c_ic_i^*, c_ic_{i+1}, c_i^*c_{i+1} \mid i \in [n]\} \cup \{a_0c_1,c_{n+1}c_{n+2}, c_{n+2}t\}$\;
    $w \gets $ $\left \{
      \begin{array}{rcl || rcll }
  a_3t   & \to & \frac{1}{b}           & c_{n+2}t       & \to & w(a_3t) + \epsilon & \\
  a_2a_3 & \to & \frac{1-w(a_3t)}{b}   & c_{n+1}c_{n+2} & \to & w(a_2a_3) - 2\epsilon - \frac{2\epsilon}{b-1}&  \\
  a_1a_2 & \to & \frac{1-w(a_2a_3)-w(a_3t)}{b} & c_ic_{i+1}     & \to & \frac{x_i \cdot w(a_1a_2) }{W} \text{~~~} \text{ for } i \in [n] & \\
  a_0a_1 & \to & \frac{1-w(a_1a_2)-w(a_2a_3)-w(a_3t)}{b}   & a_0c_1       & \to & w(a_0a_1) + \frac{2\epsilon}{b-1} & \\
  sa_0   & \to & \frac{1-w(a_0a_1)-w(a_1a_2)-w(a_2a_3)-w(a_3t)}{b} + \frac{\epsilon}{b} & & \to & 0 \text{~~~}  \text{ otherwise} & \\
  \end{array}
\right \}$\;
    \Return $(G, w, s, t, b, r=1, k=1)$
\end{algorithm}

\begin{figure}[h]
    \centering
    \begin{tikzpicture}[->,>=stealth',shorten >=1pt,auto,node distance=1.7cm,
  thick,main node/.style={circle,fill=gray!10,draw,
  font=\sffamily\bfseries,minimum size=7mm}]

  \node[main node] (s) {$s$};
  \node[main node] (a0) [right of=s] {$a_0$};
  \node[main node] (c1) [below right of=a0] {$c_1$};
  \node[main node] (a1) [above right of=a0] {$a_1$};
  \node[main node] (c2) [right of=c1] {$c_2$};
  \node[main node] (c3) [right of=c2] {$c_3$};
  \node[main node] (c4) [right of=c3] {$c_4$};
  \node[main node] (c5) [right of=c4] {$c_5$};
  \node[main node] (c1s) [below left of=c2] {$c_1^*$};
  \node[main node] (c2s) [below left of=c3] {$c_2^*$};
  \node[main node] (c3s) [below left of=c4] {$c_3^*$};
  \node[main node] (t) [above right of=c5] {$t$};
  \node[main node] (a3) [above left of=t] {$a_3$};
  \node[main node] (a2) [left of=a3] {$a_2$};

  \path[every node/.style={font=\sffamily\small,
  		fill=none,inner sep=3pt}]
  	% Right-hand-side arrows rendered from top to bottom to
  	% achieve proper rendering of labels over arrows.
    (s) edge [line width=1.7pt,bend left=0] node {$\frac{1}{32} + \frac{\epsilon}{2}$} (a0)
    (a0) edge [line width=1.7pt,bend left=20] node {$\frac{1}{16}$} (a1)
    (a1) edge [line width=1.7pt,bend left=0] node {$\frac{1}{8}$} (a2)
    (a2) edge [line width=1.7pt,bend left=0] node {$\frac{1}{4}$} (a3)
    (a3) edge [line width=1.7pt,bend left=20] node {$\frac{1}{2}$} (t)
    
    (a0) edge [line width=1.7pt,bend right=20] node {$\frac{1}{16} + 2\epsilon$} (c1)
    
    (c1) edge [line width=1.7pt, bend left=30] node {$\frac{3}{80}$} (c2)
    (c2) edge [bend left=30] node {$\frac{6}{80}$} (c3)
    (c3) edge [line width=1.7pt, bend left=30] node {$\frac{7}{80}$} (c4)
    (c1) edge [bend right=20] node {$0$} (c1s)
    (c2) edge [line width=1.7pt, bend right=20] node {$0$} (c2s)
    (c3) edge [bend right=20] node {$0$} (c3s)
    (c1s) edge [bend right=20] node {$0$} (c2)
    (c2s) edge [line width=1.7pt, bend right=20] node {$0$} (c3)
    (c3s) edge [bend right=20] node {$0$} (c4)
    
    (c4) edge [line width=1.7pt, bend right=0] node {$\frac{1}{4} - 4\epsilon$} (c5)
    (c5) edge [line width=1.7pt, bend right=20] node {$\frac{1}{2} + \epsilon$} (t)
    ;
\end{tikzpicture}
    \caption{The graph $G$ constructed by Algorithm~\ref{red:subSumToMotSGBranch} on input $X = \{3, 6, 7\}, W=10, b=2$. The solution to \subSum\,  is $X'=\{3,7\}$ and  the corresponding motivating subgraph $G'$  is formed by thick arcs.  In graph $G$ the agent is tempted to pursue the lower path $P_c=[s,a_0,c_1,c_1^*,c_2 ,c_2^*, c_3 ,c_4 ,c_5 ,t]$  but will lose motivation at node $c_5$. In graph $G'$, at node $s$, the agent's plan will be to follow the lower path $P_c$ but at node $a_0$ the agent will switch to the upper path $P_a$. At every step on  this path the agent is motivated to  reach $t$.  }
    \label{fig:subSumToMotSGBranch}
\end{figure}
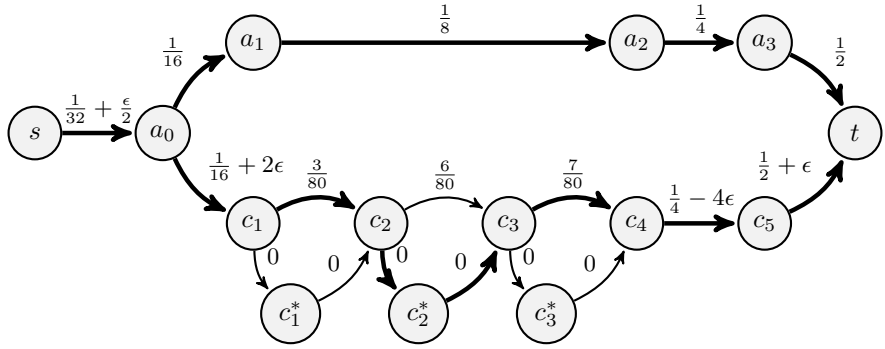

\begin{claim}
    Algorithm~\ref{red:subSumToMotSGBranch} is safe. Given as input an instance $I$ of \subSum{} and a salience factor $b \in \mathbb{R}_{> 0}$, the output instance $I'$ of \motSGbranch{} is a yes-instance if and only if $I$ is a yes-instance.
\end{claim}
\begin{proofclaim}
    Before we begin the proof, we will adapt the following notation: for a directed path $P$ and a vertex $u$ in $P$, we let $P|_u$ denote the path restricted to begin in $u$. In other words, $P|_u$ is the subgraph of $P$ induced on the vertices reachable from $u$ in $P$.
    
    For the forward direction of the proof, assume that $I$ is a yes-instance and let $X' \subseteq X$ be a witness to this. Let $G' \subseteq G$ be the graph where for each $i \in [n]$ the vertex $c_i^*$ and incident edges are removed if $x_i \in X'$, and the edge $c_ic_{i+1}$ is removed if $x_i \notin X'$. We make a series of step-wise observations which shows that $G'$ is motivating:
    \begin{enumerate}
        \item $G'$ contains a single vertex with out-degree $\geq 2$, namely $a_0$. There are two possible paths from $s$ to $t$: The $a$-path $P_a = [s, a_0, a_1, a_2, a_3, t]$; and the $c$-path $P_c = [s, a_0, c_1, (c_1^*), c_2, (c_2^*), \ldots, c_{n},(c_n^*),c_{n+1}, c_{n+2}, t]$ (where for each $i \in [n]$, $P_c$ only include vertex $c_i^*$ when $x_i \notin X'$).
        
        \item \label{enu:apathperceived} In the path $P_a$, the agent will by construction perceive the cost of moving towards $t$ to be $1$, regardless of which vertex he resides on. The only exception to this is when the agent is at $s$; then the perceived cost of following the path $P_a$ is $1 + \epsilon$. In other words, $\zeta(P_a) = 1 + \epsilon$ and for each of $i\in \{0, 1, 2, 3\}$, $\zeta(P_a|_{a_i}) = 1$.
        
        \item \label{enu:distc1cn} The distance from $c_1$ to $c_{n+1}$ in $G'$ is $\sum_{x_i \in X'}\frac{x_i \cdot w(a_1a_2)}{W}$, which simplifies to exactly $w(a_1a_2)$. This holds because $X'$ sums to $W$ by the initial assumption. In other words, $\textsc{Dist}_{P_c}(c_1, c_{n+1}) = w(a_1a_2)$.
        
        \item \label{enu:distCshorter} The distance from $a_0$ to $t$ is $\epsilon$ shorter in the $c$-path compared to the $a$-path, $\textsc{Dist}_{P_c}(a_0, t) = \textsc{Dist}_{P_a}(a_0, t) - \epsilon{}$. This follows from (\ref{enu:distc1cn}) and how the remaining weights in the $c$-path are defined.
        
        \item The perceived distance from $s$ to $t$ in $P_c$ is $1$, $\zeta(P_c) = 1$ (follows by combining (\ref{enu:apathperceived}) and (\ref{enu:distCshorter}), and observing that the two paths share their first leg). The agent is hence motivated to move from $s$ to $a_0$ with a plan of following $P_c$ towards the target.
        
        % \item An agent actually following $P_c$ will never reach the target, since the perceived cost of moving from $c_{n+2}$ to $t$ is $1+b\epsilon$, at which point the agent will lose his motivation.
        
        \item The perceived distance from $a_0$ to $t$ in $P_c$ is $\epsilon$ more than the perceived distance from $a_0$ to $t$ in $P_a$. This follows from the following chain of substitutions:
        \begin{align*}
            \zeta(P_c|_{a_0}) =\; &b\cdot w(a_0c_1) + \textsc{Dist}_{P_c}(c_1, t)  \\
            =\; &b\cdot w(a_0a_1) + b\cdot \frac{2\epsilon}{b-1} + \textsc{Dist}_{P_c}(c_1, c_{n+1}) + w(c_{n+1}c_{n+2}) + w(c_{n+2}t)  \\
            =\; &b\cdot w(a_0a_1) + b\cdot \frac{2\epsilon}{b-1} + w(a_1a_2) + w(a_2a_3) - 2\epsilon - \frac{2\epsilon}{b-1} + w(a_3t) + \epsilon  \\
            =\; &b\cdot w(a_0a_1) + w(a_1a_2) + w(a_2a_3) + w(a_3t) + (b-1)\cdot \frac{2\epsilon}{b-1} - \epsilon \\
            =\; &b\cdot w(a_0a_1) + w(a_1a_2) + w(a_2a_3) + w(a_3t) + \epsilon \\
            =\; &b\cdot w(a_0a_1) + \textsc{Dist}_{P_a}(a_1, t) + \epsilon \\
            =\; &\zeta(P_a|_{a_0}) + \epsilon
        \end{align*}
        It follows that an agent standing at $a_0$ is \emph{not} willing to walk along $P_c$, but change his plan to walking along $P_a$ instead. The agent will stay motivated on the $a$-path, and will reach the reward. Hence, the graph $G'$ is motivating.
    \end{enumerate}
    
    For the reverse direction of the proof, assume there is a motivating subgraph $G' \subseteq G$ that motivates the agent to reach $t$. We notice that $G'$ must contain the $a$-path, since any agent moving out of the $a$-path will by the construction of $G'$ need to walk past $c_{n+2}$ --- however, an agent standing at this vertex will by the construction always give up. We also notice that $P_a$ itself is not motivating, hence $G'$ must contain at least one other path from $s$ to $t$. Let $P_c$ denote the \emph{shortest} $s$-$t$ path that is different from $P_a$.
    
    We begin by observing that $P_c$ must start at $s$, include $a_0$ and $c_1,$, a path from $c_i$ to $c_{i+1}$ for every $i \in [n]$, as well as the vertices $c_{n+1}, c_{n+2}$, and $t$.
    
    There are some length requirements that $P_c$ needs to fulfill in order for $G'$ to be motivating. In order to motivate an agent at $s$ to move to $a_0$, the distance from $a_0$ to $t$ in $P_c$ can be at most $\textsc{Dist}_{P_a}(a_0, t) - \epsilon$. This implies: $$\textsc{Dist}_{P_c}(c_1, c_{n+1}) \leq w(a_1a_2)$$
    
    However, $P_c$ can not be so short that it tempts the agent to move off the $a$-path. In particular, an agent standing at $a_0$ must perceive the $a$-path to be strictly shorter than walking along $P_c$. We obtain the following inequality:
    \begin{align*}
        b \cdot w(a_0a_1) + \textsc{Dist}_{P_a}(a_1, t) &< b \cdot w(a_0c_1) + \textsc{Dist}_{P_c}(c_1, t) \\
        b \cdot w(a_0a_1) + \textsc{Dist}_{P_a}(a_1, t) &< b \cdot w(a_0a_1) + b \cdot \frac{2\epsilon}{b-1} + \textsc{Dist}_{P_c}(c_1, t) \\ 
        \textsc{Dist}_{P_a}(a_1, t) &< b \cdot \frac{2\epsilon}{b-1} + \textsc{Dist}_{P_c}(c_1, t) \\ 
        w(a_1a_2) + w(a_2a_3) + w(a_3t) &< b \cdot \frac{2\epsilon}{b-1} + \textsc{Dist}_{P_c}(c_1, c_{n+1}) + w(a_2a_3) - 2\epsilon - \frac{2\epsilon}{b-1} + w(a_3t) + \epsilon \\ 
        w(a_1a_2) &< (b - 1) \cdot \frac{2\epsilon}{b-1} + \textsc{Dist}_{P_c}(c_1, c_{n+1}) - \epsilon \\ 
        w(a_1a_2) - \epsilon &< \textsc{Dist}_{P_c}(c_1, c_{n+1}) \\ 
    \end{align*}
    We now construct a solution $X' \subseteq X$ to \subSum{} by including $x_i$ in $X'$ if $P_c$ use the edge $c_ic_{i+1}$. Notice that the sum of the edge weights for these edges will make up the distance $\textsc{Dist}_{P_c}(c_1, c_{n+1})$. We can lift the bounds for that distance to the sum of elements in $X'$:
    \begin{alignat*}{2}
        w(a_1a_2) - \epsilon &< \textsc{Dist}_{P_c}(c_1, c_{n+1}) &&\leq  w(a_1a_2) \\
        w(a_1a_2) - \epsilon &< \sum_{c_ic_{i+1} \in P_c}w(c_i, c_{i+1}) &&\leq  w(a_1a_2) \\
        w(a_1a_2) - \epsilon &< \sum_{x_i \in X'}\frac{x_i \cdot w(a_1a_2)}{W} &&\leq  w(a_1a_2) \\
        W - \epsilon \cdot \frac{W}{w(a_1a_2)} &< \sum_{x_i \in P_c} x_i &&\leq  W \\ 
    \end{alignat*}
    By choosing $\epsilon$ strictly smaller than $\frac{w(a_1a_2)}{W}$, we guarantee that the set $X'$ has value exactly $W$ (note: this bound on $\epsilon$ is a function of $b$ and $W$). This concludes the soundness proof of the reduction.
\end{proofclaim} 

By the claim,   \motSGbranch{}   is  NP-complete for $k= 1$ and hence is also  NP-complete for every $k\geq 1$. 
\end{proof}

The proof of Theroem~\ref{theoremNPdich} follow from Lemmata~\ref{lem:motPath} and~\ref{lem:motNP}.

  %%  %%%%%%%%%%%%%%%%%%%%%%%%%%%%%%%%%%%%%%%%%%%%%%%%%%%%%
   %  %%%%%%%%%%%%%%%%%%%%%%%%%%%%%%%%%%%%%%%%%%%%%%%%%%%%%
   %% %%%%%%%%%%%% END :   DICHOTOMY   %%%%%%%%%%%%%%%%%%%%
   %  %%%%%%%%%%%%%%%%%%%%%%%%%%%%%%%%%%%%%%%%%%%%%%%%%%%%%
  %%  %%%%%%%%%%%%%%%%%%%%%%%%%%%%%%%%%%%%%%%%%%%%%%%%%%%%%
  
      %%%%%%%%%%%%%%%%%%%%%%%%%%%%%%%%%%%%%%%%%%%%%%%%%%%%%  %%
      %%%%%%%%%%%%%%%%%%%%%%%%%%%%%%%%%%%%%%%%%%%%%%%%%%%%%  %
      %%%%%%%%%%%% BEGIN:  PSEUDOPOLY   %%%%%%%%%%%%%%%%%%% %%
      %%%%%%%%%%%%%%%%%%%%%%%%%%%%%%%%%%%%%%%%%%%%%%%%%%%%%  %
      %%%%%%%%%%%%%%%%%%%%%%%%%%%%%%%%%%%%%%%%%%%%%%%%%%%%%  %%
      
      %!TEX root = motivate.tex
%\section{A pseudo-polynomial XP algorithm for discrete weights}
\section{A pseudo-polynomial algorithm (proof of Theorem~\ref{theoremXP})}
\label{sec:psudopolyXP}

The reduction in Section~\ref{sec:hardbranchings} shows that restricting the number of branchings in the motivating structures we look for does not make the task of finding them significantly easier. However, we notice that the weights used in the graph constructed in the reduction can be exponentially small, and depend on $W$ as well as on $b$. On the other hand, in the hardness proof for \motSG{} by Tang et al.~\cite{tang2017computational}, the instances created in the reduction from 3-SAT have weights that depend only on the constant $b$; but the number of branchings in their potential solutions grows linearly with the number of clauses in the 3-SAT instance.

Hence, \motSG{} is hard even if the number of branchings in the solution we are looking for is bounded, and it is also hard if the input instance only use integer weights bounded by a constant. But what if we impose both restrictions simultaneously? In this section we prove that if the input instance has bounded integer weights, then we can quickly determine whether it contains a motivating subgraph with few branchings. 

%\defparproblem{\motSGbranchWeight{} (\mSGbW{})}{A time-inconsistent planning model $M = (G, w, s, t, r, b)$, and a non-negative integer $k \in \mathbb{Z}_{\geq 0}$. As an additional constraint, the edge weights are non-negative integers summing to $W$.}{$k + W$}{Does there exist a motivating subgraph $G' \subseteq G$ with at most $k$ branching vertices?}

We will make a use of an auxiliary problem which is a variation of the exact \kLinkage{} problem in acyclic digraphs, but which also impose restrictions on the links --- requiring them to be motivating for biased agents. We define the problem:

\defparproblem{\exactMotKlink{} (\textsc{EMkL})}{An acyclic digraph $G$; edge weights $w : E(G) \to \mathbb{Z}_{\geq 0}$ such that $\sum_{e \in E(G)} w(e) = W$; source terminals $s_1, s_2, \ldots, s_k \in V(G)$; sink terminals $t_1, t_2, \ldots, t_k \in V(G)$; target link weights $\ell_1, \ell_2, \ldots, \ell_k \in Z_{\geq 0}$; salience factors $b_1, b_2, \ldots, b_k \in \mathbb{R}_{\geq 1}$; and rewards $r_1, r_2, \ldots, r_k \in \mathbb{R}_{\geq 0}$.}{$k+W$}{Does there exist (internally) vertex disjoint paths $P_1, P_2, \ldots, P_k$ such that for each $i \in [k]$, $P_i$ starts in $s_i$, ends in $t_i$ and has weight $\ell_i$, and is such that an agent with salience factor $b_i$ will be motivated by a reward $r_i$ to move from $s_i$ to $t_i$?}

We solve the \exactMotKlink{} problem using dynamic programming, inspired by Fortune, Hopcroft and Wyllie's~\cite{FORTUNE1980111} solution to the \kLinkage{} problem in acyclic digraphs (see also~\cite{bang2008digraphs}).

\begin{lemma}
    \exactMotKlink{} can be solved in time $\mathcal{O}(kn^{k+1}W^k)$.
\end{lemma}
\begin{proof}
    We will in the upcoming proof assume all sources and sinks are distinct. If they are not, simply make multiple copies of each such vertex, one for each extra occurrence as a source or sink beyond the first one.

    We solve the problem by dynamic programming using a Boolean table $dp$ of size $\mathcal{O}(n^kW^k)$. The table is indexed by vertices $u_i \in V(G)$ and weights $d_i \in [W]$ for $i \in [k]$. We define a cell:
    \begin{alignat*}{3}
        dp[ u_1, u_2, \ldots, u_k, d_1, d_2, \ldots, d_k ] := &\texttt{ TRUE } &&\text{if there exists vertex disjoint paths } P_1, P_2, \ldots, P_k \\
        & &&\text{ such that for each } i \in [k] \text{ the following holds:} \\ 
        & &&\bullet P_i \text{ starts in } u_i \text{ and ends in } t_i \text{, and} \\
        & &&\bullet P_i \text{ has weight } d_i \text{, and} \\
        & &&\bullet \text{ an agent with salience factor } b_i \text{ is motivated by a reward } r_i \\
        & &&\phantom{\bullet} \text{to move from } u_i \text{ to } t_i \text{ in } P_i \text{.} \\
        & \texttt{ FALSE } &&\text{otherwise.}
    \end{alignat*}
    Observe that the final answer to the \textsc{EMkL} instance will by this definition be found in $dp[s_1, s_2, \ldots, s_k, \ell_1, \ell_2, \ldots, \ell_k]$. We proceed to establish the base case of our recurrence,  when the vertices align perfectly with the sinks. To comply with the definition, the entry is \texttt{TRUE} if the required distances are $0$, and \texttt{FALSE} otherwise.
    \begin{alignat*}{3}
        dp[ u_1, u_2, \ldots, u_k, d_1, d_2, \ldots, d_k ] \gets &\texttt{ TRUE } &&\text{if for every } i \in [k], u_i = t_i \text{ and } d_i = 0 \text{.} \\
        & \texttt{ FALSE } &&\text{if for every } i \in [k], u_i = t_i \text{ and } d_i \neq 0 \text{.}
    \end{alignat*}
    We move on to describe the recurrence. The idea is to move one step forward, trying every neighbor of each vertex in the current state, to see whether it is possible to find a slightly shorter partial solution we can extend.
    \begin{alignat*}{3}
        dp[ u_1, u_2, \ldots, u_k, d_1, d_2, \ldots, d_k ] \gets &\texttt{ TRUE } &&\text{if there exists } i \in [k] \text{ and } v \in N(u_i) \text{ such that:} \\
        & &&\bullet \text{ for all $j \in[k] \setminus \{i\}$, there is no path from $u_j$ to $u_i$ in $G$, and} \\
        & &&\bullet \text{ for all $j \in[k]$, $v \neq u_j$} \\
        & &&\bullet \text{ $(b_i - 1)\cdot w(u_iv) + d_i \leq r$, and} \\
        & &&\bullet \text{ $dp[u_1, \ldots, u_{i-1}, v, u_{i+1}, \ldots, d_{i-1}, d_i - w(u_iv), d_{i+1}, \ldots, d_k] = \texttt{TRUE}$.} \\
        & \texttt{ FALSE } &&\text{otherwise. }
    \end{alignat*}
    Calculating the recurrence can be done in time $\mathcal{O}(kn)$ if efficient data structures are used for storing the sets of reachable vertices for each vertex. %, and the set of vertices used to index the table.
    Before we prove the correctness of the recurrence, note that we can safely assume that no source contains any in-edges, and no sink contains any out-edges; otherwise we do some simple prepossessing to ensure this holds. We can further assume all sinks come at the end of a topological sort of the graph.
    
    We begin by proving the forward direction. Assume there exist vertex disjoint paths $P_1, P_2, \ldots, P_k$ which satisfy the required conditions. For a path $P_i$ and a positive integer $j_i \leq |V(P)|$, we let $P_i^{j_i}$ denote the tail of $P_i$ containing the $j_i$ last vertices. Further, let $s_i^{j_i}$ be the first vertex of $P_i^{j_i}$, and let $d_i^{j_i}$ be the weight of $P_i^{j_i}$. We claim that for every combination of $j_i$ for distinct $i \in [k]$, it holds that $dp[s_1^{j_1}, s_2^{j_2}, \ldots, s_k^{j_k}, d_1^{j_1}, d_2^{j_2}, \ldots, d_k^{j_k}] = \texttt{TRUE}$.
    
    We prove the claim by induction on the sum of the lengths of the tails, $J = \sum_{i \in [k]} j_i$. The base case occurs at $J = k$, when every tail has length $1$. The induction hypothesis when proving the claim for larger $J$ will be that the claim holds for $J-1$, and from there we work our way up to the case when $J = \sum_{i \in [k]}|V(P_i)|$, at which point the forward direction of the correctness proof is complete.
    
    In the base case $J = k$, we observe that $j_i = 1$ for every $i \in [k]$, since the domains for the $j_i$'s are subsets of strictly positive integers, and their sum is $k$. Hence, $s_i^{j_i} = t_i$ and $d_i^{j_i} = 0$. By the base case of the recurrence, it then holds that $dp[s_1^{j_1}, s_2^{j_2}, \ldots, s_k^{j_k}, d_1^{j_1}, d_2^{j_2}, \ldots, d_k^{j_k}] = \texttt{TRUE}$.
    
    For the inductive step, consider some combination of $j_i$ for all distinct $i \in [k]$ whose sum is $J$. Among the first vertices of the corresponding tails, pick the one who comes first according to some topological order, say $s_i^{j_i}$. Recall that we can assume all sinks are last in the topological order, hence $s_i^{j_i}$ is not a sink and $P_i^{j_i}$ contains at least two vertices. Let $v$ be the second vertex of $P_i^{j_i}$. We observe that all the bullet points for the recurrence to give \texttt{TRUE} is satisfied when we try $i$ and $v$: Since $s_i^{j_i}$ is the earliest in the topological order, no other start vertices has a path to it. By virtue of $P_1, P_2, \ldots, P_k$ being a valid solution, we know $v$ is not in another path, and we know that the agent with salience factor $b_i$ is willing to move along the edge $s_i^{j_i}v$ when the distance from $v$ to $t_i$ is $d_i^{j_{i}} - w(s_{i}^{j_i}v)$. And by the induction hypothesis, we know that $dp[s_1^{j_1}, \ldots, s_{i-1}^{j_{i-1}}, v, s_{i+1}^{j_{i-1}}, \ldots, d_{i-1}^{j_{i-1}}, d_i^{j_{i}} - w(s_{i}^{j_i}v), d_{i+1}^{j_{i+1}}, \ldots, d_k^{j_k}] = \texttt{TRUE}$. This concludes the proof for the forward direction.
    
    For the backward direction, assume $dp[s_1, s_2, \ldots, s_k, \ell_1, \ell_2, \ldots, \ell_k] = \texttt{TRUE}$. We retrieve a solution as follows. Initially, let $P_i = {s_i}$ for each $i \in [k]$. Next, we iteratively append vertices to the paths in the following manner: Starting with $u_j \gets s_j$ and $d_j \gets \ell_j$ for $j \in [k]$, let $i$ and $v$ be a choice in the recurrence that satisfied the bullet point requirements for $dp[u_1, u_2, \ldots, u_k, d_1, d_2, \ldots, d_k]$. Append $v$ to $P_i$, and let $d_i \gets d_i - w(u_iv)$ and $u_i \gets v$ in the next iteration. Continue until we are left with only sinks.
    
    We show that the paths created are disjoint. Assume for the sake of contradiction that two paths $P_i$ and $P_j$ are not disjoint, and let $v$ be the topologically first vertex they share. Let $v'_i$ and $v'_j$ be the predecessors of $v$ in respectively $P_i$ and $P_j$. Without loss of generality, assume $v$ was added to $P_i$ before it was added to $P_j$. When $v$ was added to $j$, it was done so when $u_i = v'_i$. But at the time, the vertex $u_j$ must have had a path to $v$ as well (through $v'_j$), hence the first condition of the recurrence is not satisfied. This is a contradiction.
    
    We observe that the distance of $P_i$ is $\ell_i$, and that by the requirements in the recurrence an agent will always be motivated to move along the path. Thus, the constructed paths do indeed form a solution. This concludes the proof.
\end{proof}

Returning to our problem \motSGbranch{} with integer weights, we make use of the following observations that enables us to give an algorithm.

\begin{proposition}[{\cite[Theorem~5.1]{Kleinberg2018timeinconsistent}}] \label{prop:motsubgraphs_are_sparse}
    If $G'$ is a minimal motivating subgraph, %(\ie no proper subgraphs of $G'$ are motivating),
    then it contains a unique $s$-$t$ path $P \subseteq G'$ that the agent will follow. Moreover, every node of $G'$ has at most one outgoing edge that does not lie on $P$. 
\end{proposition}

Note that a consequence of Proposition~\ref{prop:motsubgraphs_are_sparse} is that all branching vertices of a minimal motivating subgraph are on the path $P$.
% \begin{observation} \label{obs:branchings_le2_outedges}
%     Consider a yes-instance of \motSGbranchWeight{}, and let $G' \subseteq G$ be a minimal solution witnessing this (\ie no proper subgraphs of $G'$ are motivating). Let $P$ be the path taken by the agent in $G'$. Then all branching vertices of $G'$ are in $P$, and no branching vertex has more than two out-edges.
% \end{observation}

\begin{definition}[Merging vertex]
    The {\em in-degree} of a vertex $v$ in a directed graph $G$ is the number of edges in $G$ that have $v$ as its second endpoint. We  say that $v$ is a \emph{merging vertex}, if its in-degree is at least two.
\end{definition}

\begin{observation} \label{obs:mergings_lek}
    Consider an instance of the time-inconsistent planning model $M = (G, w, s, t, r, b)$, and let $G' \subseteq G$ be a minimal motivating subgraph with $k$ branchings. Then there are at most $k$ merging vertices in $G'$.
\end{observation}

We now give the algorithm for \motSGbranch{} with integer weights, where we use the algorithm for \exactMotKlink{} above as a subroutine. In short, the approach is to first guess which $\mathcal{O}(k)$ vertices that are ``interesting'' in the solution, \ie have either in-degree or out-degree (or both) larger than $1$ in $G'$, and then try every possible way of connecting these points together using disjoint paths. Finally, we guess how heavy each such path should be, and apply the algorithm for \exactMotKlink{}. For details, see Algorithm~\ref{alg:MSGBW}.

\begin{algo}[\motSGbranch{} with integer weights] \label{alg:MSGBW}
    {\em Input:} An instance $(G, w, s, t, b, r, k)$ of \motSGbranch{}, whose edge weights are integer that sum to $W$. {\em Output:} \texttt{TRUE} if there exists a motivating subgraph $G' \subseteq G$ with at most %
    %\footnote{While this algorithm answers the exact version, it can simply be run on increasing values of $k$ to answer the ``at most'' version for the cost of an extra factor $k$ in the runtime.} 
    $k$ branchings, \texttt{FALSE} otherwise.
    
    %Before delving into the core of the algorithm, recall that we can preprocess $G$ such that every vertex can be reached from $s$, and every vertex can reach $t$. We can furthermore assume that $s$ has out-degree exactly $1$; otherwise we  add a single vertex with a 0-weight edge to $s$ and let this be the start vertex.
    We assume that $k \geq 1$, as otherwise we apply Algorithm~\ref{alg:motpath}. We also assume that there exists no motivating subgraph with strictly less than $k$ branching vertices; if we are unsure, we first run this same algorithm with parameter $k-1$ (this will cause an extra factor $k$ in the runtime).
    
    We begin the algorithm by guessing the ``interesting'' vertices (in addition to $s$ and $t$) of the subgraph we are looking for, $G'$. We guess:
    \begin{itemize}
        \item A set $B$ of $k$ distinct branching vertices; we let $B = \{u_1, u_2,\ldots,u_k\} \subseteq V(G)$ such that each selected vertex has out-degree at least two in $G$. Vertices are named according to their unique topological order --- if there is no unique such order, skip this iteration (this is safe, since if $B$ does not have a unique topological order, then there can not exist a path that visit all of $B$).
        \item A set $B^\star$ of $k$ distinct ``next-step'' vertices; these are the immediate next vertices our agent will go to when standing at a branching vertex in $G'$. We let $B^\star = \{u_1^\star,u_2^\star,\ldots,u_k^\star\} \subseteq V(G)$ such that for each $i \in [k]$, $u_i^\star$ is in the out-neighborhood of $u_i$. Furthermore, for each $i \in [k-1]$, $u^\star_i$ must have a path to $u_{i+1}$ (unless $u^\star_i = u_{i+1}$). 
        \item A set $B^\diamond$ of $k$ (not necessarily distinct) ``shortcut'' vertices; these are vertices the agent will {\em not} choose to go to from a branching vertex in $G'$. We let $B^\diamond = \{u^\diamond_1,u^\diamond_2,\ldots,u^\diamond_k\} \subseteq V(G)$ be such that for each $i \in [k]$, $u^\diamond_i$ is in the out-neighborhood of $u_i$, yet is different from $u^\star_i$.
        \item A set $B^\succ$ of $k$ (not necessarily distinct) merge vertices; these are the vertices of $G'$ with in-degree at least two. We let $B^\succ = \{u^\succ_1, u^\succ_2, \ldots, u^\succ_k\} \subseteq V(G)$.
    \end{itemize}
    While it is at this point mostly clear how the agent's path $P$ will move through our set of vertices (it will follow the unique topological order of $B \cup B^\star$) we still need to guess how the shortcuts behave, and in particular how the merge vertices interact.
    
    Towards this purpose, we create the (unweighted) {\em reachability graph} $H$, which illustrates every possible way of interconnecting the interesting vertices of $G'$ we just guessed. Let $V(H) = B \cup B^\star \cup B^\diamond \cup B^\succ \cup \{s, t\}$, and let there be an edge if one vertex is reachable from another in $G$ without going through vertices of $H$ --- in other words, let $E(H) = \{uv \in V(H) \times V(H) \mid \text{ there is a path from }u\text{ to } v \text{ in }G[(V(G)\setminus V(H)) \cup \{u, v\}]\}$.
    
    We are now ready to guess exactly how the interesting points are interconnected in $G'$. We will do this by guessing $H'$, a graph which can be obtained from $G'$ (if it exists) by repeatedly smoothing non-interesting vertices. {\em Smoothing} a vertex is possible when it has both in-degree and out-degree exactly $1$, and entails replacing the vertex by an edge whose weigh is the sum of the two previous edge weights. In particular, we guess:
    \begin{itemize}
        \item A subgraph $H' \subseteq H$ that obeys the following constraints:%
        \begin{itemize}
            \item $V(H') = V(H)$.
            \item For each $i \in [k]$, the out-neighbors of $u_i$ are exactly $u^\star_i$ and $u^\diamond_i$.
            \item For each non-branching vertex $u \in V(H') \setminus (B \cup \{t\})$, it has exactly one out-neighbor.
            \item There exists an $s$-$t$ path $P'$ that contains the edge $u_iu^\star_i$ for every $i \in [k]$.
        \end{itemize}
        \item A weight function $w' : E(H') \to \mathbb{Z}_{\geq 0}$ that obeys the following constraints:
        \begin{itemize}
            \item  The sum of edges is bounded by $W$, i.\,e.~$\sum_{e \in E(H')}w'(e) \leq W$.
            \item The weight function $w'$ is consistent with $w$: for each $e \in E(H')\cap E(G)$, $w'(e) = w(e)$.
            \item For each branching vertex $u_i$, an agent standing there must evaluate the path along $P'$ to be strictly better than moving to $u^\diamond_i$. More formally, for each $i \in [k]$, it holds that $b \cdot w'(u_iu^\star_i) + \textsc{Dist}_{H'}(u^\star_i, t) < b \cdot w'(u_iu^\diamond_i) + \textsc{Dist}_{H'}(u^\diamond_i, t)$.
        \end{itemize} 
    \end{itemize}

    We now create an instance $I$ of \exactMotKlink{}:
    \begin{itemize}
        \item Let $G$ be the graph.
        \item Let $w$ be the weight function.
        \item Let the source terminals be the startpoints of edges $E(H')$.
        \item Let the sink terminals be the endpoints of edges $E(H')$.
        \item Let the target distances be the weight of edges in $E(H')$ according to $w'$.
        \item Salience factor: If the corresponding edge $uv \in E(H')$ is also in $P'$, then let the salience factor be $b$. Otherwise, let it be $1$.
        \item Reward: If the corresponding edge $uv \in E(H')$ is also in $P'$, let the reward be $r$ minus the shortest distance from $v$ to $t$ in $H'$. Otherwise, let the reward equal to the target distance.
    \end{itemize}

    If $I$ is a yes-instance of \exactMotKlink, we return \texttt{TRUE}. If all created $I$-instances across all guesses are no-instances, we return \texttt{FALSE}.
\end{algo}

\begin{proof}[{Proof of Theorem~\ref{theoremXP}}]
    Recall that the theorem states that \motSGbranch{} can be solved in time $(|V(G)| \cdot W)^{\mathcal{O}(k)}$. We first prove that Algorithm~\ref{alg:MSGBW} is correct, and return to runtime at the end of the proof.
    
    Observe that the pre-processing steps do not change the answer. For the forward direction, assume $G' \subseteq G$ is a minimal solution witnessing that the input is a yes-instance with the minimum possible number of branching vertices, and let $P \subseteq G'$ be the path taken by the agent in $G'$. Let $k$ be the number of branchings in $G'$ --- we assume $k \geq 1$, as otherwise Algorithm~\ref{alg:motpath} would suffice. Let $B = \{u \in V(G') \mid |N_{G'}(u)| \geq 2\}$. Since vertices of $B$ all have out-degree at least $2$ in $G$, this set $B$ will be guessed by Algorithm~\ref{alg:MSGBW}.
    
    Similarly, let $B^\star$ be the set of vertices which are immediate successors to a vertex of $B$ in $P$, let $B^\diamond$ be the neighbors of vertices of $B$ that are not an immediate successor of that same vertex in $P$, and let $B^\succ$ be the set of vertices of in-degree at least two in $G'$. Due to Observation~\ref{obs:mergings_lek} it holds that $|B^\succ| \leq k$, so we see that Algorithm~\ref{alg:MSGBW} will at some point guess all these sets correctly.
    
    Imagine that we create the reachability graph $H'$ by repeatedly smoothing all vertices of $G'$ that have both in-degree and out-degree exactly $1$, unless the vertex is in one of the sets $B^\star, B^\diamond$. In the smoothing process of a vertex $u$, we give the resulting edge weight equal to the sum of edge weights previously incident to $u$. By the minimality of $G'$, the graph that remains is on the vertex set $B \cup B^\star \cup B^\diamond \cup M \cup \{s, t\}$. The number of edges in $H'$ is exactly $V(H') + k - 1$, since every vertex has one out-edge, except the $k$ branching vertices, which have two each, and the sink, which has none. Hence, $H'$ and its weight function will be guessed by Algorithm~\ref{alg:MSGBW}.
    
    Conducting the same smoothing process on $P$ will yield the path $P'$ found by the algorithm as well. By virtue of $G'$ being a valid solution with $P$ being the path taken by the agent, the iteration where all of the above is guessed correctly will not be skipped due to some branch tempting the agent to walk off the path. We observe that the resulting instance of \exactMotKlink{} is a yes-instance --- the smoothed segments of $G'$ to obtain $H'$ are internally vertex disjoint and have length equal to the corresponding edge in $H'$. For the segments with motivation requirements, that is, edges in $P'$ corresponding to path segments of $P$, we know that they satisfy the motivation requirement because they do so in $G'$.
    
    For the backward direction, assume the algorithm returned \texttt{TRUE}, and let $B$, $B^\star$, $B^\diamond$, $B^\succ$, $H'$, and $w'$ be the guesses that led to this conclusion, let $P'$ be the found path in $H'$, and let $I$ be the created yes-instance of \exactMotKlink{}. We build a solution $G'$ by gluing together the vertices of $H$ with the paths witnessing that $I$ is a yes-instance, and let $P$ denote the expansion of $P'$ in $G'$. The agent will never be tempted to walk away from $P$, since the next step of $P$ always appears like the best option at all branching points in every guess of $w'$. It remains to show that the agent is indeed motivated to move at all, for every edge of $P$.
    
    Every edge $uv \in E(P)$ is part of some segment found in the solution to $I$. Let $s'$ and $t'$ be the source and sink in $I$ where $uv$ was part of the solution. Since $uv$ is in $P$, the edge $s't'$ was in $P'$, and as such the corresponding salience factor $b'$, was set to $b$ in $I$. Since $I$ is a yes-instance, the agent is indeed motivated to move along every segment between $s'$ and $t'$, including across $uv$.
    
    Finally, we argue for the runtime. The number of ways to guess $B$, $B^\star$, $B^\diamond$, and $B^\succ$, is $\mathcal{O}(n^{4k})$. Notice that the number of edges in $H'$ is exactly $|V(H')| + k' - 1 \leq 5k + 1$. However, the out-edges of $B$ and $t$ are not up for guessing, and the other vertices each have exactly one out-edge each in $H'$. The number of ways to guess $H'$ is thus $\mathcal{O}((4k)^{3k})$. For the weights, $2k$ of them are already fixed because they exist in $E(G)$, so it remains to guess the weights for at most $3k + 1$ of them. The number of ways to guess $w'$ is bounded by $\mathcal{O}(W^{3k + 1})$. Checking validity of weight functions takes $\mathcal{O}(k)$ and constructing $I$ takes $\mathcal{O}(n+m)$, but these will be dominated by the $\mathcal{O}(kn^{5k+2}W^{5k+1})$ time spent to solve the \exactMotKlink{} instance. Including the extra factor $k$ for incrementally trying larger values of $k$, gives a total runtime of $\mathcal{O}(kn^{4k}(4k)^{3k}W^{3k + 1}kn^{5k+2}W^{5k+1})$, which simplifies to $(nW)^{\mathcal{O}(k)}$.
\end{proof}
  
  %%  %%%%%%%%%%%%%%%%%%%%%%%%%%%%%%%%%%%%%%%%%%%%%%%%%%%%%
   %  %%%%%%%%%%%%%%%%%%%%%%%%%%%%%%%%%%%%%%%%%%%%%%%%%%%%%
   %% %%%%%%%%%%%% END :   PSEUDOPOLY   %%%%%%%%%%%%%%%%%%%
   %  %%%%%%%%%%%%%%%%%%%%%%%%%%%%%%%%%%%%%%%%%%%%%%%%%%%%%
  %%  %%%%%%%%%%%%%%%%%%%%%%%%%%%%%%%%%%%%%%%%%%%%%%%%%%%%%

      %%%%%%%%%%%%%%%%%%%%%%%%%%%%%%%%%%%%%%%%%%%%%%%%%%%%%  %%
      %%%%%%%%%%%%%%%%%%%%%%%%%%%%%%%%%%%%%%%%%%%%%%%%%%%%%  %
      %%%%%%%%%%%% BEGIN:  CONCLUSION   %%%%%%%%%%%%%%%%%%% %%
      %%%%%%%%%%%%%%%%%%%%%%%%%%%%%%%%%%%%%%%%%%%%%%%%%%%%%  %
      %%%%%%%%%%%%%%%%%%%%%%%%%%%%%%%%%%%%%%%%%%%%%%%%%%%%%  %%
      
      \section{Proof of Theorem~\ref{theoremW1}}
      In this section we make the observation that \motSGbranch{} can not be solved in time $(|V(G)|\cdot W)^{\mathcal{O}(1)} \cdot f(k)$ for any function $f$ of $k$ only, unless FPT=W[1].
      
      It follows from Slivkins~\cite{Slivkins10} that the \kLinkage{} problem is W[1]-hard parameterized by $\ell$. In~\cite{albers2019motivating}, Albers and Kraft give a reduction from the \kLinkage{} problem to \motSG{} where all feasible minimal motivating subgraphs have exactly $\ell$ branching vertices. Moreover, the edge weights in their reduction --- when scaled to integers --- are bounded by a polynomial function of $\ell$. It thus follows that a run-time of $(|V(G)|\cdot W)^{\mathcal{O}(1)} \cdot f(k)$ for \motSGbranch{} would place the \kLinkage{} problem in FPT.

      \section{Conclusion and further work}
      
       We have shown that the \motSGbranch{} problem is polynomial-time solvable when $k=0$, and NP-complete otherwise (Theorem~\ref{theoremNPdich}). However, when edge weights are (scaled to) integers and their sum is bounded by $W$, we gave a pseudo-polynomial algorithm which solves the problem in time $(|V(G)|\cdot W)^{\mathcal{O}(k)}$ (Theorem~\ref{theoremXP}). Finally, we observed that an algorithm with run-time $(|V(G)|\cdot W)^{\mathcal{O}(1)} \cdot f(k)$ where $f$ is a function of $k$ only is not possible unless FPT=W[1] (Theorem~\ref{theoremW1}).
      
      We end the paper with an open question and a reflection for further research. First, a question; in Theorem~\ref{theoremW1}, a more careful analysis reveals that the statement holds even when $W \in \mathcal{O}(k^2)$. But is it also the case when $W$ is a constant?

      Finally, a reflection. In Theorem~\ref{theoremXP} we give a pseudo-polynomial algorithm for \motSGbranch{} when the edge weights are scaled to integer values and the number of branchings is constant. While we can reasonably assume that edge weights are represented as fractions or integers after storing them in a computer, this suggests that the values might have been quantized or approximated in some way. However, this rounding could potentially alter the solution space. In particular when the cost of two paths are perceived to be almost equal at the branching point --- then rounding values even a tiny bit can have big consequences. To overcome this, one might change the model such that the perceived difference of costs must exceed an epsilon for the agent's choice to be unequivocal.

\bibliographystyle{unsrt}  
\bibliography{references}  %%% Remove comment to use the external .bib file (using bibtex).

\begin{thebibliography}{10}

\bibitem{Kleinberg2018timeinconsistent}
Jon Kleinberg and Sigal Oren.
\newblock Time-inconsistent planning: A computational problem in behavioral
  economics.
\newblock {\em Communications of the ACM}, 61(3):99--107, February 2018.

\bibitem{tang2017computational}
Pingzhong Tang, Yifeng Teng, Zihe Wang, Shenke Xiao, and Yichong Xu.
\newblock Computational issues in time-inconsistent planning.
\newblock In {\em Proceedings of the 31st t AAAI Conference on Artificial
  Intelligence (AAAI)}, 2017.

\bibitem{albers2019motivating}
Susanne Albers and Dennis Kraft.
\newblock Motivating time-inconsistent agents: A computational approach.
\newblock {\em Theory of Computing Systems}, 63(3):466--487, 2019.

\bibitem{Samuelson1937}
Paul~A. Samuelson.
\newblock A note on measurement of utility.
\newblock {\em The Review of Economic Studies}, 4(2):155--161, 02 1937.

\bibitem{Laibson1994}
David~I. Laibson.
\newblock {\em Hyperbolic Discounting and Consumption}.
\newblock PhD thesis, Massachusetts Institute of Technology, Department of
  Economics, 1994.

\bibitem{akerlof1991procrastination}
George~A. Akerlof.
\newblock Procrastination and obedience.
\newblock {\em The American Economic Review}, 81(2):1--19, 1991.

\bibitem{McClure2004}
Samuel~M. McClure, David~I. Laibson, George Loewenstein, and Jonathan~D. Cohen.
\newblock Separate neural systems value immediate and delayed monetary rewards.
\newblock {\em Science}, 306(5695):503--507, 2004.

\bibitem{Frederick2002}
Shane Frederick, George Loewenstein, and Ted O'Donoghue.
\newblock Time discounting and time preference: A critical review.
\newblock {\em Journal of Economic Literature}, 40(2):351--401, 2002.

\bibitem{o1999doing}
Ted O'Donoghue and Matthew Rabin.
\newblock Doing it now or later.
\newblock {\em American Economic Review}, 89(1):103--124, 1999.

\bibitem{albers2017penalties}
Susanne Albers and Dennis Kraft.
\newblock On the value of penalties in time-inconsistent planning.
\newblock In Ioannis Chatzigiannakis, Piotr Indyk, Fabian Kuhn, and Anca
  Muscholl, editors, {\em 44th International Colloquium on Automata, Languages,
  and Programming (ICALP 2017)}, volume~80 of {\em Leibniz International
  Proceedings in Informatics (LIPIcs)}, pages 10:1--10:12, Dagstuhl, Germany,
  2017. Schloss Dagstuhl--Leibniz-Zentrum fuer Informatik.

\bibitem{kleinberg2016planning}
Jon Kleinberg, Sigal Oren, and Manish Raghavan.
\newblock Planning problems for sophisticated agents with present bias.
\newblock In {\em Proceedings of the 2016 ACM Conference on Economics and
  Computation}, pages 343--360. ACM, 2016.

\bibitem{kleinberg2017planning}
Jon Kleinberg, Sigal Oren, and Manish Raghavan.
\newblock Planning with multiple biases.
\newblock In {\em Proceedings of the 2017 ACM Conference on Economics and
  Computation}, pages 567--584. ACM, 2017.

\bibitem{albers2017price}
Susanne Albers and Dennis Kraft.
\newblock The price of uncertainty in present-biased planning.
\newblock In {\em International Conference on Web and Internet Economics},
  pages 325--339. Springer, 2017.

\bibitem{gravin2016procrastination}
Nick Gravin, Nicole Immorlica, Brendan Lucier, and Emmanouil Pountourakis.
\newblock Procrastination with variable present bias.
\newblock In {\em Proceedings of the 2016 ACM Conference on Economics and
  Computation}, pages 361--361. ACM, 2016.

\bibitem{Slivkins10}
Aleksandrs Slivkins.
\newblock Parameterized tractability of edge-disjoint paths on directed acyclic
  graphs.
\newblock {\em {SIAM} Journal on Discrete Mathematics}, 24(1):146--157, 2010.

\bibitem{bang2008digraphs}
J.~Bang-Jensen and G.Z. Gutin.
\newblock {\em Digraphs: Theory, Algorithms and Applications}.
\newblock Springer Monographs in Mathematics. Springer London, 2008.

\bibitem{Karp72}
Richard~M. Karp.
\newblock Reducibility among combinatorial problems.
\newblock In {\em Complexity of Computer Computations}, pages 85--103. Plenum
  Press, New York, 1972.

\bibitem{FORTUNE1980111}
Steven Fortune, John Hopcroft, and James Wyllie.
\newblock The directed subgraph homeomorphism problem.
\newblock {\em Theoretical Computer Science}, 10(2):111 -- 121, 1980.

\end{thebibliography}


% George Kour and Raid Saabne.
% \newblock Real-time segmentation of on-line handwritten arabic script.
% \newblock In {\em Frontiers in Handwriting Recognition (ICFHR), 2014 14th
%   International Conference on}, pages 417--422. IEEE, 2014.

% George Kour and Raid Saabne.
% \newblock Fast classification of handwritten on-line arabic characters.
% \newblock In {\em Soft Computing and Pattern Recognition (SoCPaR), 2014 6th
%   International Conference of}, pages 312--318. IEEE, 2014.

% \bibitem{hadash2018estimate}
% Guy Hadash, Einat Kermany, Boaz Carmeli, Ofer Lavi, George Kour, and Alon
% \newblock Estimate and replace: A novel approach to integrating deep neural
%   networks with existing applications.
% \newblock {\em arXiv preprint arXiv:1804.09028}, 2018.

% \end{thebibliography}

%%  %%%%%%%%%%%%%%%%%%%%%%%%%%%%%%%%%%%%%%%%%%%%%%%%%%%%%
 %  %%%%%%%%%%%%%%%%%%%%%%%%%%%%%%%%%%%%%%%%%%%%%%%%%%%%%
 %% %%%%%%%% END:     BIBLIOGRAPHY   %%%%%%%%%%%%%%%%%%%%
 %  %%%%%%%%%%%%%%%%%%%%%%%%%%%%%%%%%%%%%%%%%%%%%%%%%%%%%
%%  %%%%%%%%%%%%%%%%%%%%%%%%%%%%%%%%%%%%%%%%%%%%%%%%%%%%%

%\appendix
%\section{Appendix}

\end{document}